\documentclass[12pt]{article}

\usepackage{geometry}
\usepackage{cite}
\usepackage{graphicx}
\usepackage{subfigure}
\usepackage{xspace}
\usepackage{amssymb}
\usepackage{amsmath}
\allowdisplaybreaks[1]

\usepackage{amsthm}

\newtheorem{theorem}{Theorem}
\newtheorem{lemma}[theorem]{Lemma}

\newtheorem{observation}[theorem]{Observation}

\makeatletter
\def\@endtheorem{\endtrivlist}
\makeatother

\newcounter{rules}
\newenvironment{Rule}{\refstepcounter{rules}\par\smallskip\noindent
\textbf{(\arabic{rules})}\quad}{} 
\newcommand{\currentrule}{\arabic{rules}}

\newcommand{\algb}{\textsc{D4PVC}}
\newcommand{\alg}[3]{\algb(#1,#2,#3)}

\newcommand{\vc}{$V_1$-disjoint cover\xspace}
\newcommand{\mvc}{minimum $V_1$-disjoint cover\xspace}
\newcommand{\component}[1]{$V_{#1}$-component}
\newcommand{\components}[1]{$V_{#1}$-components}

\newcommand{\sboundary}{S^x_b}
\newcommand{\scontained}{S^x_c}

\begin{document}

\title{An $O^*(2.619^k)$ algorithm for 4-path vertex cover}
\author{Dekel Tsur%
\thanks{
Ben-Gurion University of the Negev.
Email: \texttt{dekelts@cs.bgu.ac.il}}}
\date{}
\maketitle

\begin{abstract}
In the \emph{$4$-path vertex cover problem},
the input is an undirected graph $G$ and an integer $k$.
The goal is to decide whether there is a set of vertices $S$ of
size at most $k$ such that every path with $4$ vertices in $G$ contains at
least one vertex of $S$.
In this paper we give a parameterized algorithm for $4$-path vertex cover
whose time complexity is $O^*(2.619^k)$.
\end{abstract}


\section{Introduction}
For an undirected graph $G$, an \emph{$l$-path} is a path in $G$ with $l$
vertices.
An \emph{$l$-path vertex cover} is a set of vertices $S$ such that every 
$l$-path in $G$ contains at least one vertex of $S$.
In the \emph{$l$-path vertex cover problem},
the input is an undirected graph $G$ and an integer $k$. The goal is to decide
whether there is an $l$-path vertex cover of $G$ with size at most $k$.
The problem for $l=2$ is the famous vertex cover problem.
For every fixed $l$, there is a simple reduction from vertex cover
to $l$-path vertex cover.
Therefore, $l$-path vertex cover is NP-hard for every constant $l \geq 2$.

For every fixed $l$, the $l$-path vertex cover problem has a simple
parameterized algorithm with running time $O^*(l^k)$~\cite{cai1996fixed}.
It is possible to obtain better algorithms for specific values of $l$.
Faster algorithms for 3-path vertex cover were given
in~\cite{tu2015fixed,wu2015measure,katrenivc2016faster,chang2016fixed,xiao2017kernelization,tsur2018parameterized}.
The currently fastest algorithm for 3-path vertex cover has $O^*(1.713^k)$
running time~\cite{tsur2018parameterized}.
For the 4-path vertex cover problem, Tu et al.\ gave an $O^*(3^k)$-time
algorithm~\cite{tu2016fpt}.
{\v{C}}erven{\`y} gave an $O^*(4^k)$-time algorithm for 5-path vertex
cover~\cite{cerveny2018}.

In this paper we give an algorithm for 4-path vertex cover whose time
complexity is $O^*(2.619^k)$.

\section{Preliminaries}
For a graph $G=(V,E)$ and a vertex $v\in V$,
$N(v)$ is the set of vertices that are adjacent to $v$ and $\deg(v) = |N(v)|$.
Additionally, for a set of vertices $S$,
$N(S) = (\bigcup_{v\in S} N(v))\setminus S$.
For set of vertices $S$, $G[S]$ is the subgraph
of $G$ induced by $S$ (namely, $G[S]=(S,E\cap (S\times S))$).
We also define $G-S = G[V\setminus S]$.

In order to give an algorithm for 4-path vertex cover, we use the
\emph{iterative compression method} (cf.~\cite{cygan2015parameterized}).
Consider the following problem called \emph{disjoint 4-path vertex cover}.
The input is an undirected graph $G$, an integer $k$, and a 4-path
vertex cover $V_1$ of $G$. 
The goal is to decide whether there is a 4-path vertex cover $S$ of $G$
such that $S\cap V_1 = \emptyset$ and $|S| \leq k$.
If there is an $O^*(c^k)$ algorithm for the disjoint 4-path vertex cover problem
then there is an $O^*((c+1)^k)$ algorithm for 4-path vertex cover.
Therefore, in the rest of the paper we will describe an algorithm for disjoint 
4-path vertex cover with running time $O^*(1.619^k)$.
We will assume that for an instance $G,V_1,k$ of the disjoint 4-path vertex
cover problem, $G[V_1]$ does not contain a 4-path since in this case the
instance is trivially a no instance.

A set of vertices $S$ will be called a \emph{\vc} of $G$ if $S$ is a 4-path
vertex cover of $G$ and $S \cap V_1 = \emptyset$.
Let $V_2 = V\setminus V_1$.
Denote $N_i(v) = N(v) \cap V_i$, $\deg_i(v) = |N_i(v)|$,
and $N_i(S) = N(S) \cap V_i$.
A vertex $x\in V_1$ is called a \emph{connection vertex} if $\deg_1(x) = 0$
and $\deg_2(x) \geq 2$.
A vertex $x\in V_1$ is called a \emph{leaf} if $\deg_1(x) = 0$
and $\deg_2(x) = 1$.
A connected component of $G[V_i]$ is called a \emph{\component{i}}.
We say that a \component{2} $C$ is \emph{contained} in a connection vertex $x$
if $C \subseteq N(x)$.
We say that a connection vertex $x$ \emph{splits} a \component{2} $C$ if
$C \cap N(x) \neq \emptyset$ and $C \not\subseteq N(x)$.

A set of vertices $I$ is an \emph{independent set} if there is no edge
between two vertices of $I$. Note that a set $I$ of size~1 is an independent
set.
A graph $G=(V,E)$ is called a \emph{star} if $|V| \geq 3$ and
there is a vertex $v\in V$ such that $v$ is adjacent to all the other vertices
in the graphs and $V\setminus \{v\}$ is an independent set.
The vertex $v$ is called the \emph{center} of the star.

\section{The algorithm}
In this section we describe a branching algorithm \algb\ for solving
the disjoint 4-path vertex cover problem.
The input to the algorithm is a graph $G$, a 4-path vertex cover $V_1$ of $G$,
and an integer $k$.
We note that the size of $V_1$ can be larger than $k$.

Let $S_1,\ldots,S_t$ be subsets of $V_2$.
We say that the algorithm \emph{recurses on} $S_1,\ldots,S_t$ if for each $S_i$,
the algorithm tries to find a solution $S$ that contains $S_i$.
More precisely, the algorithm performs the following lines.
\begin{enumerate}
\item
For $i=1,\ldots,t$:
\begin{enumerate}
\item
If $\alg{G-S_i}{V_1}{k-|S_i|}$ returns `yes', return `yes'.
\end{enumerate}
\item
Return `no'.
\end{enumerate}

We now describe the reduction and branching rules of the algorithm.
The algorithm applies the first applicable rule from the following rules.
For most of the branching rules below, the branching vector is $(1,t)$ for
$t\geq 2$ and therefore the branching number is at most 1.619.
\begin{Rule}\label{xrule1}
If $k < 0$ or $k = 0$ and $G$ contains a 4-path return `no'.
\end{Rule}

\begin{Rule}\label{xrule2}
If $G$ does not contain a 4-path return `yes'.
\end{Rule}

\begin{Rule}\label{rule:no-4-path}
If $C$ is a connected component of $G$ that does not contain a 4-path return
$\alg{G-C}{V_1 \setminus C}{k}$.
\end{Rule}

\begin{Rule}\label{rule:small-component}
If $C$ is a connected component of $G$ such that $|C \cap V_2| \leq 3$, find
a \mvc $S_1$ of $G[C]$ (by enumerating all subsets of $C \cap V_2$).
Recurse on $S_1$.
\end{Rule}

\begin{Rule}\label{rule:move-to-V1}
If there is a vertex $v \in V_2$ such that either
(1) $\deg_2(v)$ = 1 and either $\deg_1(v) = 0$ or
all the vertices in $N_1(v)$ are leaves, or
(2) the \component{2} of $v$ is a triangle and $\deg_1(v) = 0$,
return $\alg{G}{V_1 \cup \{v\}}{k}$.
See Figure~\ref{fig:rule-move-to-V1}.
\end{Rule}
\begin{figure}
\centering
\includegraphics{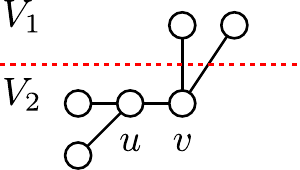}
\caption{An example for Rule~(\ref{rule:move-to-V1}).}
\label{fig:rule-move-to-V1}
\end{figure}
\begin{lemma}
Rule~(\currentrule) is correct.
\end{lemma}
\begin{proof}
To prove the lemma we show that there is a \mvc of $G$ that does not
contain $v$.
Let $S$ be a \mvc of $G$ and suppose that $v\in S$.
If $\deg_2(v) = 1$, let $u$ be the single vertex in $N_2(v)$.
Then, $S' = (S \setminus \{v\}) \cup \{u\}$ is also a minimum 4-path vertex
cover of $G$: Suppose conversely that $G-S'$ contains a 4-path.
This path must contain $v$. However, the connected component of $v$ in $G-S'$
consists of $v$ and its adjacent leaves, and thus this component does not
contain a 4-path, a contradiction.

Next consider the case in which the \component{2} $C$ of $v$ is a triangle.
There is a vertex $v_1 \in C$ such that $v_1 \notin S$
(otherwise, $S\setminus \{v\}$ is a \vc of $G$, contradicting
the assumption that $S$ is a \mvc of $G$).
Let $v_2$ be the single vertex in $C \setminus \{v,v_1\}$.
Then, $S' = (S \setminus \{v\}) \cup \{v_1\}$ is also a minimum 4-path vertex
cover of $G$: Conversely, if $G-S'$ contains a 4-path, it must contain $v$,
and therefore it must be of the form $a,b,v_2,v$
(since $N(v) = \{v_1,v_2\}$ and $v_1 \in S'$).
Then, $a,b,v_2,v_1$ is a 4-path in $G-S$, contradicting
the assumption that $S$ is a \vc of $G$.
\end{proof}

\begin{Rule}\label{rule:P4}
If there is a 4-path $P$ in $G$ such that $|P \cap V_2| = 1$,
recurse on $\{v\}$, where $v$ is the single vertex in $P \cap V_2$.
See Figure~\ref{fig:P4}.
\end{Rule}
\begin{figure}
\centering
\includegraphics{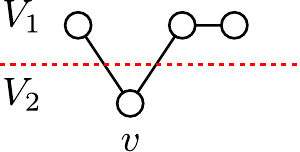}
\caption{An example for Rule~(\ref{rule:P4}).}
\label{fig:P4}
\end{figure}


\begin{Rule}\label{rule:P3}
If there is a path $P = x_1,x_2,x_3$ in $G$ such that
$|P \cap V_2| = 1$ and $|N_2(\{x_1, x_3\}) \setminus P| \geq 2$,
recurse on $\{v\}$ and $N_2(\{x_1, x_3\}) \setminus P$,
where $v$ is the single vertex in $P \cap V_2$.
See Figure~\ref{fig:P3}.
\end{Rule}
\begin{figure}
\centering
\includegraphics{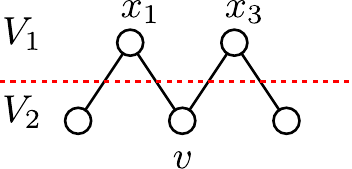}
\caption{An example for Rule~(\ref{rule:P3}).}
\label{fig:P3}
\end{figure}
\begin{lemma}
Rule~(\currentrule) is correct.
\end{lemma}
\begin{proof}
If $S$ is a \mvc of $G$ such that $v \notin S$,
then for every $u\in N_2(x_1) \setminus P$,
$u \in S$ (since $S$ needs to cover the path $u,x_1,x_2,x_3$).
Similarly, $N_2(x_3) \setminus P \subseteq S$.
\end{proof}
\begin{observation}\label{obs:P3}
If Rules (1)--(\currentrule) cannot be applied and $x$ is a connection vertex
such that there is a vertex $v \in N(x)$ with $N_1(v) \neq \{x\}$,
then $\deg_2(x) = 2$.
\end{observation}
\begin{proof}
If $\deg_2(x) \geq 3$ then Rule~(\ref{rule:P3}) can be applied on the path
$x,v,y$ where $y$ is a vertex in $N_1(v) \setminus \{x\}$.
\end{proof}


\begin{Rule}\label{rule:xy}
If there is a \component{1} $C$ of size at least~2,
let $v \in N(C)$ be a vertex such that $\deg_2(v) = 1$,
and let $u$ be the single vertex in $N_2(v)$. Recurse on $\{u\}$.
See Figure~\ref{fig:xy}.
\end{Rule}
\begin{figure}
\centering
\includegraphics{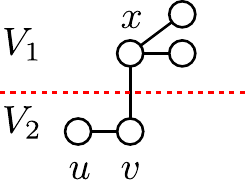}
\caption{An example for Rule~(\ref{rule:xy}).}
\label{fig:xy}
\end{figure}
\begin{lemma}
Rule~(\currentrule) is correct.
\end{lemma}
\begin{proof}
First note that due to Rule~(\ref{rule:no-4-path}) and since we assumed that
$G[V_1]$ does not contain a 4-path, at least one vertex of $C$
has a neighbor in $V_2$.
Since Rule~(\ref{rule:P4}) cannot be applied, $C$ cannot be a triangle.
Additionally, if $C$ is a star then the only vertex in $C$ with neighbors
in $V_2$ is the center of $C$.

For every vertex $w \in N_2(C)$ we have that
\begin{enumerate}
\item
$N_1(w) \subseteq C$ (since Rule~(\ref{rule:P4}) cannot be applied).
\item
$\deg_2(w) \leq 1$ (otherwise Rule~(\ref{rule:P3}) can be applied on a path
$w,x,y$ for $x,y \in C$).
\end{enumerate}
We next show that the vertex $v$ exists.
Suppose conversely that $\deg_2(w) \neq 1$ for every $w \in N(C)$.
By property~2 above, $\deg_2(w) = 0$ for every $w \in N_2(C)$.
From property~1 we obtain that the connected $C'$ in $G$ that contains $C$
is $C \cup N_2(C)$.
Since Rule~(\ref{rule:small-component}) cannot be applied,
$|N_2(C)| \geq 4$.
Since Rule~(\ref{rule:no-4-path}) cannot be applied, $C'$ contains a 4-path.
Therefore, $C$ cannot be a star (otherwise $C'$ is also a star and thus
$C$ does not contain a 4-path), and thus $|C| = 2$.
Denote the vertices of $C$ by $x$ and  $y$.
We have that $\deg_2(x) \geq 1$ and $\deg_2(y) \geq 1$ (otherwise $C'$ does not
contain a 4-path).
This implies that, without loss of generality, there is a vertex $w \in N_2(x)$
such that $|N_2(y) \setminus \{w\}| \geq 2$.
Therefore, Rule~(\ref{rule:P3}) can be applied on the path $w,x,y$,
a contradiction.
Thus, there is a vertex $v \in N_2(C)$ such that $\deg_2(v) = 1$.
Let $x \in C$ be a vertex adjacent to $v$.

Let $S$ be a \mvc of $G$ and suppose that $u \notin S$.
Due to the path $u,v,x,y$, where $y$ is some vertex from $C\setminus \{x\}$,
we have that $v \in S$.
Define $S' = (S \setminus \{v\}) \cup \{u\}$.
We will show that $S'$ is a \mvc of $G$.
Since $|S'| \leq |S|$, it suffices to show that $S'$ is a \vc of $G$.

We consider two cases. In the first case assume that $\deg_2(y) = 0$ for every
$y\in C \setminus \{x\}$.
Therefore, $N(v) = \{x,u\}$.
Suppose conversely that $G-S'$ contains a 4-path.
This path must contain $v$, and since $N(v) = \{x,u\}$ and $u \in S'$,
it follows that the path is of the form $a,b,x,v$, where
$a,b \notin C$ (as $N(y) = \{x\}$ for every $y \in C\setminus \{x\}$).
Therefore, $a,b,x,y$ is a 4-path in $G-S$,
where $y$ is some vertex from $C\setminus \{x\}$, a contradiction.
Therefore, $S'$ is a \vc of $G$.

In the second case suppose that $\deg_2(y) \geq 1$
for some $y \in C\setminus \{x\}$.
Since $x,y \in C$ have neighbors in $V_2$, $C$ cannot be a star, so
it follows that $C = \{x,y\}$.
We claim that $N_2(x) = N_2(y) = \{v\}$.
If $N_2(y) \setminus \{u, v\} \neq \emptyset$, Rule~(\ref{rule:P3}) can be
applied on the path $v,x,y$, a contradiction.
Thus, $N_2(y) \subseteq \{u,v\}$.
Similarly, $N_2(x) \subseteq \{u,v\}$
(otherwise, Rule~(\ref{rule:P3}) can be applied on the path $v',y,x$, where
$v' \in N_1(y)$, a contradiction).
If $y$ is adjacent to $u$ then by property~2 above
$\deg_2(u) \leq 1$ and therefore $N_2(u) = \{v\}$.
We also have that $N_1(u) \subseteq \{x,y\}$ and $N_1(v) \subseteq \{x,y\}$
(property~1).
Thus, the connected component of $x$ is $x,y,u,v$, contradicting the
assumption that Rule~(\ref{rule:small-component}) cannot be applied.
It follows that $N_2(y) = \{v\}$.
Due to symmetry, we also obtain that $N_2(x) = \{v\}$.

Suppose conversely that $G-S'$ contains a 4-path.
This path must contain $v$. However, the connected component of $v$ in $G-S'$
is $v,x,y$ and it does not contain a 4-path, a contradiction.
Therefore, $S'$ is a \vc of $G$.
\end{proof}
\begin{observation}
If Rules (1)--(\currentrule) cannot be applied, $V_1$ is an independent set.
\end{observation}

\begin{Rule}\label{rule:two-connection-vertices}
If a vertex $v \in V_2$ is adjacent to two connection vertices $x_1,x_2$,
return $\alg{G'}{V_1}{k}$, where $G'$ is the graph obtained from $G$
by deleting the edge $(v,x_2)$.
See Figure~\ref{fig:two-connection-vertices}.
\end{Rule}
\begin{figure}
\centering
\includegraphics{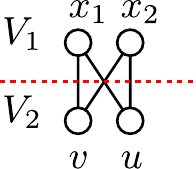}
\caption{An example for Rule~(\ref{rule:two-connection-vertices}).}
\label{fig:two-connection-vertices}
\end{figure}
\begin{lemma}
Rule~(\currentrule) is correct.
\end{lemma}
\begin{proof}
Since Rule~(\ref{rule:P3}) cannot be applied, there is a vertex $u \in V_2$
such that $N(x_1) = N(x_2) = \{v,u\}$
(otherwise, $|N(\{x_1,x_2\}) \setminus \{v\}| \geq 2$, so
Rule~(\ref{rule:P3}) can be applied on the path $x_1,v,x_2$).
We now show that a set of vertices $S$ is a \vc of $G$ if and only if
$S$ is a \vc of $G'$.
Since $G'$ is a subgraph of $G$,
if $S$ is a \vc of $G$ then $S$ is also a \vc of $G'$.
To show the second direction of the claim, let $S$ be a \vc of $G'$.
Suppose conversely that $S$ is not a \vc of $G$, namely, there
is a 4-path $P$ in $G-S$. The path $P$ must pass through the edge $(v,x_2)$.
Therefore, $v \notin S$.
Since $S$ must cover the path $v,x_1,u,x_2$, we have $u \in S$.
The fact that $N(x_1) = N(x_2) = \{v,u\}$ implies that the path $P$ is of the
form $a,b,v,x_2$, where $a \neq x_1$ and $b \neq x_1$.
Therefore, $a,b,v,x_1$ is a 4-path in $G-S'$, a contradiction.
We conclude that $S$ is a \vc of $G$.
\end{proof}
\begin{observation}\label{obs:two-connection-vertices}
If Rules (1)--(\currentrule) cannot be applied,
every vertex $v\in V_2$ is adjacent to at most one connection vertex.
\end{observation}

\begin{Rule}\label{rule:boundary-vertex}
If $x$ is a connection vertex that splits a \component{2} $C$
such that there is a vertex $u \in C \cap N(x)$ for which
$|N_2(u) \setminus N(x)| \geq 2$, recurse on $\{u\}$ and
$N_2(u) \setminus N(x)$.
See Figure~\ref{fig:boundary-vertex}.
\end{Rule}
\begin{figure}
\centering
\includegraphics{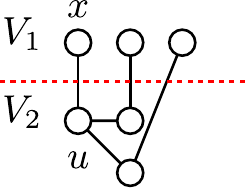}
\caption{An example for Rule~(\ref{rule:boundary-vertex}).}
\label{fig:boundary-vertex}
\end{figure}
\begin{lemma}
Rule~(\currentrule) is correct.
\end{lemma}
\begin{proof}
Let $S$ be a \vc of $G$. If $u \in S$ we are done.
Otherwise, we will show that $N_2(u) \setminus N(x) \subseteq S$.
Fix $v \in N_2(u) \setminus N(x)$.
Every \component{2} of size at least~3 is either a star or a triange.
Therefore, $C$ is either a star whose center is $u$ or a triangle.
In the former case we have that $\deg_2(v) = 1$.
Since Rule~(\ref{rule:move-to-V1}) cannot be applied, $\deg_1(v) \geq 1$.
Choose $y \in N_1(v)$. By definition, $v$ is not adjacent to $x$,
hence $x \neq y$.
The set $S$ must contain a vertex of the path $x,u,v,y$. Due to the assumption
that $u \notin S$, we conclude that $v \in S$.
Since this is true for every $v \in N_2(u) \setminus N(x)$, we obtain that
$N_2(u) \setminus N(x) \subseteq S$.
\end{proof}

If Rules (1)--(\currentrule) cannot be applied and $x$ is a connection vertex
that splits a \component{2} $C$, there is a unique vertex $v \in C$ such
that $v \notin N(x)$ and $N_2(v) \cap N(x) \neq \emptyset$.
This vertex will be called the \emph{boundary vertex} of $C$
with respect to $x$.

\begin{Rule}\label{rule:contained-connected-component}
If $C$ is a \component{2} of size at least 3 that is contained in
a connection vertex $x$,
choose $v_1,v_2,v_3\in C$ such that $v_1,v_2,v_3$ is a path
and recurse on $\{v_2\}$.
See Figure~\ref{fig:contained-connected-component}.
\end{Rule}
\begin{figure}
\centering
\includegraphics{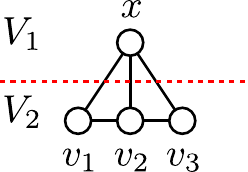}
\caption{An example for Rule~(\ref{rule:contained-connected-component}).}
\label{fig:contained-connected-component}
\end{figure}
\begin{lemma}
Rule~(\currentrule) is correct.
\end{lemma}
\begin{proof}
Let $S$ be a \mvc of $G$ and suppose that $v_2 \notin S$.
Due to the path $x,v_1,v_2,v_3$, $S$ contains either $v_1$ or $v_3$.
Without loss of generality assume that $v_1 \in S$.
Then, $S' = (S \setminus \{v_1\}) \cup \{v_2\}$ is a \mvc of $G$:
Suppose conversely that there is a 4-path $P$ in $G-S'$.
From Observation~\ref{obs:P3}, $N_1(v_i) = \{x\}$ for all $i$.
The \component{2} $C$ is either a star whose center is $v_2$ or a
triangle.
Therefore, $P$ is of the form $a,b,x,v_1$, $a,b,x,v_3$, or $a,x,v_1,v_3$
(and $a,b \notin \{v_1,v_3\}$).
In the former case $a,b,x,v_2$ is a path in $G-S$ and in the latter case
$a,x,v_2,v_3$ is a path in $G-S$, contradicting the assumption that $S$
is a \vc of $G$.
Therefore, $S'$ is a \mvc of $G$.
\end{proof}

\begin{Rule}\label{rule:triangle}
If $C$ is a \component{2} which is a triangle,
let $x$ be a connection vertex that splits $C$,
and let $v$ be the boundary vertex of $C$ with respect to $x$.
Recurse on $\{v\}$ and $C \setminus \{v\}$.
See Figure~\ref{fig:triangle}.
\end{Rule}
\begin{figure}
\centering
\includegraphics{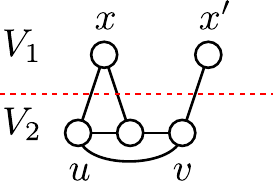}
\caption{An example for Rule~(\ref{rule:triangle}).}
\label{fig:triangle}
\end{figure}
\begin{lemma}
Rule~(\currentrule) is correct.
\end{lemma}
\begin{proof}
We first show that $x$ exists.
Since Rule~(\ref{rule:small-component}) cannot be applied, there is
a vertex in $C$ that is adjacent to a connection vertex $x$.
The vertex $x$ splits $C$ (otherwise, $C$ is contained in $x$, contradicting
the assumption that Rule~(\ref{rule:contained-connected-component}) cannot be
applied).
Therefore, $x$ exists.
Since Rule~(\ref{rule:boundary-vertex}) cannot be applied, $|C \cap N(x)| = 2$.
Denote $C\cap N(x) = \{u, u'\}$.

Since Rule~(\ref{rule:move-to-V1}) cannot be applied, $v$ is adjacent to a
vertex $x' \in V_1$. By definition, $v$ is not adjacent to $x$ and therefore
$x \neq x'$.
Let $S$ be a \mvc of $G$ and suppose that $v \notin S$.
Due to the path $x,u,v,x'$, $S$ must contain $u$.
Similarly, $u'\in S$. Therefore, $C\setminus \{v\} \subseteq S$.
\end{proof}

\begin{Rule}\label{rule:split-1-leaves}
If $x$ is a connection vertex that splits exactly one \component{2} $C$
and at least one of the vertices in $N(x)$ is adjacent to a leaf,
recurse on $\{u\}$,
where $u \in N(x)$ is a vertex that is adjacent to the boundary
vertex of $C$ with respect to $x$.
See Figure~\ref{fig:split-1-leaves}.
\end{Rule}
\begin{figure}
\centering
\includegraphics{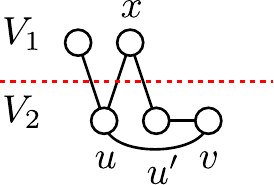}
\caption{An example for Rule~(\ref{rule:split-1-leaves}).}
\label{fig:split-1-leaves}
\end{figure}
\begin{lemma}
Rule~(\currentrule) is correct.
\end{lemma}
\begin{proof}
By Observation~\ref{obs:P3}, $\deg(x) = 2$.
Denote $N(x) = \{u, u'\}$ and let $v$ be the boundary vertex of $C$ with respect
to $x$.
From Observation~\ref{obs:two-connection-vertices}, $u$ and $u'$ are not
adjacent to connection vertices other than $x$.
Let $S$ be a \mvc of $G$ and suppose that $u \notin S$.
There is a 4-path in $G$ whose vertices are $u,x,u'$ and a leaf that is
adjacent to $u$ or $u'$.
Therefore, $u' \in S$.
We claim that the set $S' = (S \setminus \{u'\}) \cup \{u\}$ is a \mvc of $G$.
Assume conversely that $G-S'$ contains a 4-path.
This path must contain $u'$, so it is of the form $a,b,v,u'$ or $a,v,u',y$,
where $y \in N_1(u')$ (namely, $y$ is either $x$ or a leaf adjacent to $u'$).
In the former case $a,b,v,u$ is a 4-path in $G-S$, and in the latter case
$a,v,u,x$ is a 4-path in $G-S$.
This is a contradiction to the assumption that $S$ is a \vc.
Therefore, $S'$ is a \mvc of $G$.
\end{proof}

\begin{Rule}\label{rule:notIS}
If $x$ is a connection vertex that splits a \component{2} $C$ such that
$C\cap N(x)$ is not an independent set, recurse on $\{u\}$ and
$\{v\} \cup (N(x) \setminus C)$, where $u$ is the center of $C$
and $v$ is the boundary vertex of $C$ with respect to $x$.
See Figure~\ref{fig:notIS}.
\end{Rule}
\begin{figure}
\centering
\includegraphics{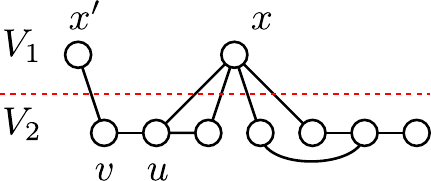}
\caption{An example for Rule~(\ref{rule:notIS}).}
\label{fig:notIS}
\end{figure}
\begin{lemma}\label{lem:notIS}
Rule~(\currentrule) is correct.
\end{lemma}
\begin{proof}
If $N_1(w) \neq \{x\}$ for some $w\in N(x)$ then by Observation~\ref{obs:P3},
$\deg_2(x) = 2$.
Since $|C\cap N(x)| \geq 2$, it follows that $C$ is the only \component{2} that
$x$ splits, and thus Rule~(\ref{rule:split-1-leaves}) can be applied,
a contradiction.
Therefore, $N_1(w) = \{x\}$ for every $w\in N(x)$.

Since Rule~(\ref{rule:move-to-V1}) cannot be applied, $\deg_1(v) \geq 1$.
Let $x'$ be a vertex in $N_1(v)$.
Let $S$ be a \mvc of $G$ and assume that $u \notin S$.
Due to the path $x',v,u,x$, $S$ must contain $v$.
If $S \cap (C \cap N(x)) = \emptyset$ then $N(x) \setminus C \subseteq S$.
Otherwise, $S' = (S \setminus (C \cap N(x))) \cup \{u\}$ is a \mvc of $G$:
Conversely, if $G-S'$ contains a 4-path, this path must be of the form
$a,b,x,u'$ for $u' \in (C \cap N(x)) \setminus \{u\}$
($a,b \notin C \cap N(x)$).
This implies that $G-S$ contains a path $a,b,x,u$, a contradiction.
\end{proof}

For a connection vertex $x$,
we denote by $\sboundary$ the set containing the boundary vertex of $C$ with
respect to $x$ for every \component{2} $C$ that $x$ splits.
See Figure~\ref{fig:Sb} for an example.
\begin{figure}
\centering
\includegraphics{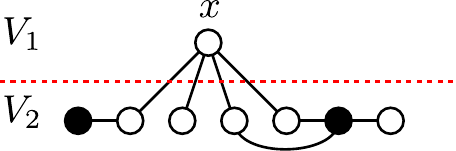}
\caption{An example for the definition of $\sboundary$.
The vertices of $\sboundary$ are marked in black.}
\label{fig:Sb}
\end{figure}

\begin{Rule}\label{rule:contains-special}
If $x$ is a connection vertex such that $|C\cap N(x)| = 1$ for every
\component{2} that $x$ splits, $x$ contains exactly one \component{2} $C'$,
and $|C'| = 2$, recurse on $N(x) \setminus C'$ and $\sboundary \cup \{u\}$,
where $u$ is a vertex in $C'$.
See Figure~\ref{fig:contains-special}.
\end{Rule}
\begin{figure}
\centering
\includegraphics{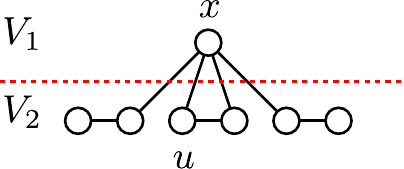}
\caption{An example for Rule~(\ref{rule:contains-special}).}
\label{fig:contains-special}
\end{figure}
\begin{lemma}
Rule~(\currentrule) is correct.
\end{lemma}
\begin{proof}
The assumption that $C'$ is the only \component{2} that is contained in $x$
implies that $x$ splits at least one \component{2} (otherwise,
by Observation~\ref{obs:two-connection-vertices} we have that
Rule~(\ref{rule:small-component}) can be applied).
Therefore, $\deg(x) \geq 3$.
By Observation~\ref{obs:P3}, $N_1(v) = \{x\}$ for every $v\in N(x)$.

Let $S$ be a \mvc of $G$.
If $S\cap C' = \emptyset$ then $N(x) \setminus C' \subseteq S$
(for every $v \in N(x)\setminus C'$ there is a 4-path $v,x,u,u'$ in $G$,
where $u'$ is the single vertex in $C'\setminus\{u\}$, and therefore $v\in S$)
and we are done.

Now assume that $S\cap C' \neq \emptyset$.
Let $S_0 = N(x) \cup \sboundary$ and
$S' = (S \setminus S_0) \cup \sboundary \cup \{u\}$.
The set $S'$ is a \vc of $G$: Assume conversely that $G-S'$ contains a 4-path.
This path must contain a vertex 
$v \in S_0 \setminus (\sboundary \cup \{u\}) = N(x) \setminus \{u\}$.
However, the connected component
of $v$ in $G-S'$ is a star (recall that $N_1(w) = \{x\}$ for every $w\in N(x)$),
a contradiction.
We will show that $|S\cap S_0|\geq |\sboundary \cup \{u\}|$ and therefore
$S'$ is a \mvc of $G$.

Let $t\geq 1$ be the number of \component{2} that $x$ splits.
By definition, $|\sboundary \cup \{u\}| = t+1$.
If $|S\cap C'| = 1$ then for every \component{2} $C$ that $x$ splits,
$S$ contains at least one vertex from $S_0 \cap C$
(namely, $S$ contains either the single vertex in $C\cap N(x)$ or
the boundary vertex of $C$ with respect to $x$).
Therefore, $|S\cap S_0| \geq t+1$.
Otherwise (if $|S \cap C'| = 2$), for every \component{2} $C$ that $x$ splits
except at most one,
$S$ contains at least one vertex from $S_0 \cap C$.
Therefore, $|S\cap S_0| \geq (t-1)+2 = t+1$.
\end{proof}
The branching vector of Rule~(\currentrule) is $(t,t+1)$, where $t \geq 1$.
Therefore, the branching number is at most 1.619.

\begin{Rule}\label{rule:contains}
If $x$ is a connection vertex that contains at least one \component{2},
recurse on $\sboundary \cup \scontained$, where
$\scontained$ is a set containing one vertex from each \component{2} of size~2
that is contained in $x$.
See Figure~\ref{fig:contains}.
\end{Rule}
\begin{figure}
\centering
\includegraphics{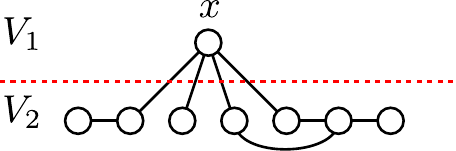}
\caption{An example for Rule~(\ref{rule:contains}).
In this example, $s_1 = 1$, $s_2 = 0$, $t_1 = 1$, and $t_2 = 1$.}
\label{fig:contains}
\end{figure}
\begin{lemma}
Rule~(\currentrule) is correct.
\end{lemma}
\begin{proof}
We first claim that $N_1(v) = \{x\}$ for every $v \in N(x)$.
Suppose conversely that $N_1(v) \neq \{x\}$ for some $v \in N(x)$.
By Observation~\ref{obs:P3}, $\deg(x) = 2$.
If $x$ does not split \components{2}, then due to
Observation~\ref{obs:two-connection-vertices},
the connected component of $x$ in $G$ consists of $\{x\} \cup N(x)$ and leaves
adjacent to vertices in $N(x)$.
Therefore, Rule~(\ref{rule:small-component}) can be applied, a contradiction.
Thus, $x$ splits at least one \component{2}.
Since $\deg(x) = 2$ and $x$ contains at least one \component{2}, it follows
that $x$ splits exactly one \component{2}.
By Observation~\ref{obs:two-connection-vertices}, $v$ is adjacent to a leaf.
Therefore, Rule~(\ref{rule:split-1-leaves}) can be applied, a contradiction.
Hence, $N_1(v) = \{x\}$ for every $v \in N(x)$.

Let $S$ be a \mvc of $G$.
Let $S_0 = N(x) \cup \sboundary$.
The set
$S' = (S\setminus S_0) \cup (\sboundary \cup \scontained)$
is a \vc of $G$: If $G-S'$ contains a 4-path then this path must contain
a vertex $v\in N(x)$. However, the connected component of $v$ in $G-S'$ is
a star, a contradiction.
We will show that
$|S\cap S_0| \geq |\sboundary \cup \scontained |$
and therefore $S'$ is a \mvc of $G$.

Let $s_1$ (resp., $s_2$) be the number of \components{2} of size~1
(resp., size~2) that are contained in $x$.
Let $t_1$ (resp., $t_2$) be the number of \components{2} $C$
that $x$ splits and $|C\cap N(x)| = 1$ (resp., $|C\cap N(x)| \geq 2$).
By definition, $|\sboundary \cup \scontained| = s_2+t_1+t_2$.

Suppose that $s_2 = 0$. Then, from the assumption that $x$ contains at least
one \component{2}, $s_1 \geq 1$.
If for every \component{2} $C$ that $x$ splits, $S$ contains at least one
vertex from $S_0 \cap C$
then $|S\cap S_0| \geq t_1+t_2 = s_2+t_1+t_2$.
Otherwise, let $C$ be a \component{2} that $x$ splits such that
$S$ does not contain a vertex from $S_0 \cap C$.
It follows that $N(x) \setminus C \subseteq S$.
Therefore,
\[ |S\cap S_0| \geq |N(x) \setminus C| \geq s_1+(t_1+t_2-1) \geq t_1+t_2
= s_2+t_1+t_2.
\]
Now suppose that $s_2 \geq 1$. Since Rule~(\ref{rule:contains-special}) cannot
be applied, either $s_1 \geq 1$, $s_2 \geq 2$, or $t_2 \geq 1$.
If there is a \component{2} $C$ of size~2 that is contained in $x$ and
$C\cap S = \emptyset$ then $N(x)\setminus C \subseteq S$ and thus
$|S\cap S_0| \geq s_1+2(s_2 -1)+t_1+2t_2 \geq s_2+t_1+t_2$,
where the last inequality follow from the assumption that $s_2 \geq 1$ and
from the fact that either $s_1 \geq 1$, $s_2 \geq 2$, or $t_2 \geq 1$.
Otherwise, $S$ contains at least one vertex from each \component{2} of size~2
that is contained in $x$.
If there is a \component{2} $C$ of size~2 that is contained in $x$ such that
$|C\cap S| = 1$,
then for every \component{2} $C'$ that $x$ splits,
$S$ contains at least one vertex from $S_0 \cap C'$.
Therefore, $|S\cap S_0| \geq s_2+t_1+t_2$.
Otherwise (if $S$ contains the vertices of every \component{2} of size~2 that is
contained in $x$),
for every \component{2} $C'$ that $x$ splits except at most one,
$S$ contains at least one vertex from $S_0 \cap C'$.
Therefore, $|S\cap S_0| \geq 2s_2+(t_1+t_2-1) \geq s_2+t_1+t_2$.
\end{proof}

\begin{Rule}\label{rule:split-1}
If $x$ is a connection vertex that splits exactly one \component{2} $C$,
recurse on $\sboundary$.
See Figure~\ref{fig:split-1}.
\end{Rule}
\begin{figure}
\centering
\includegraphics{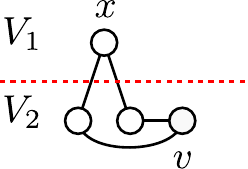}
\caption{An example for Rule~(\ref{rule:split-1}).}
\label{fig:split-1}
\end{figure}
\begin{lemma}
Rule~(\currentrule) is correct.
\end{lemma}
\begin{proof}
By Observation~\ref{obs:two-connection-vertices} and
since Rule~(\ref{rule:split-1-leaves}) cannot be applied,
$N_1(u) = \{x\}$ for every $u \in N(x)$.
Additionally, $x$ does not contain \components{2}
(due to Rule~(\ref{rule:contains})).
Denote $\sboundary = \{v\}$.
Let $S$ be a \mvc of $G$ and assume that $v \notin S$.
$S$ must contain at least one vertex from $N(x)$.
Therefore, the set $S' = (S\setminus N(x)) \cup \{v\}$ is a \mvc of $G$.
\end{proof}
\begin{observation}\label{obs:split-1}
If Rules (1)--(\currentrule) cannot be applied,
every connection vertex $x$ splits at least two \components{2}.
\end{observation}

\begin{Rule}\label{rule:degv-leaf}
If $x$ is a connection vertex that splits a \component{2} $C$ such that the
boundary vertex $v$ of $C$ with respect to $x$ satisfies $\deg_1(v) \geq 1$
and at least one of the vertices in $N(x)$ is adjacent to a leaf,
recurse on $\{u\}$ and $\{v\} \cup (N(x) \setminus \{u\})$ where $u$ is the
single vertex in $C \cap N(x)$.
See Figure~\ref{fig:degv-leaf}.
\end{Rule}
\begin{figure}
\centering
\includegraphics{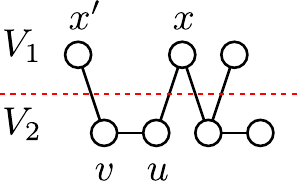}
\caption{An example for Rule~(\ref{rule:degv-leaf}).}
\label{fig:degv-leaf}
\end{figure}
\begin{lemma}
Rule~(\currentrule) is correct.
\end{lemma}
\begin{proof}
By Observation~\ref{obs:P3}, $\deg(x) = 2$.
From Observation~\ref{obs:split-1} we obtain that $|C \cap N(x)| = 1$.
Let $x'$ be a vertex in $N_1(v)$.
By definition, $v$ is not adjacent to $x$ and therefore $x \neq x'$.
Let $S$ be a \mvc of $G$ and assume that $u \notin S$.
Due to the path $x',v,u,x$ in $G$, $S$ must contain $v$.
Since at least one of the two vertices in $N(x)$ is adjacent to a leaf,
$S$ must contain the single vertex in $N(x) \setminus \{u\}$.
\end{proof}

\begin{Rule}\label{rule:degv}
If $x$ is a connection vertex that splits a \component{2} $C$ such that the
boundary vertex $v$ of $C$ with respect to $x$ satisfies $\deg_1(v) \geq 1$,
recurse on $\{u\}$ and $\sboundary$,
where $u$ is a vertex in $C \cap N(x)$,
See Figure~\ref{fig:degv}.
\end{Rule}
\begin{figure}
\centering
\includegraphics{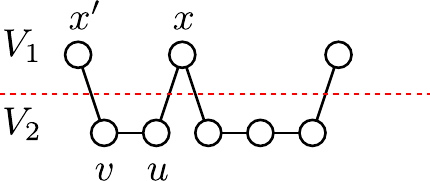}
\caption{An example for Rule~(\ref{rule:degv}).}
\label{fig:degv}
\end{figure}
\begin{lemma}\label{lem:degv}
Rule~(\currentrule) is correct.
\end{lemma}
\begin{proof}
Let $x'$ be a vertex in $N_1(v)$.
Let $S$ be a \mvc of $G$ and assume that $u \notin S$.
Due to the path $x',v,u,x$, $S$ must contain $v$.
Let $S_0 = N(x) \cup \sboundary$ and
$S' = (S\setminus S_0) \cup \sboundary$.
$S'$ is a vertex cover of $G$.
Additionally, for every \component{2} $C'$ that $x$ splits,
$S$ must contain at least one vertex from $S_0 \cap C$.
Therefore $|S\cap S_0| \geq |\sboundary|$.
Therefore, $S'$ is a \mvc of $G$.
\end{proof}
From Rule~(\ref{rule:move-to-V1}) and Rule~(\currentrule) we obtain the
following observation.
\begin{observation}\label{obs:degv}
If Rules (1)--(\currentrule) cannot be applied and $C$ is a \component{2} that
is split by a connection vertex $x$, then $C$ is a star whose center $v$
is the boundary vertex of $C$ with respect to $x$.
Additionally, $\deg_1(v) = 0$.
\end{observation}

\begin{Rule}\label{rule:large-intersection}
If $x$ is a connection vertex that splits a \component{2} $C$ such that
$|C \cap N(x)|\geq 2$, recurse on $\sboundary$.
See Figure~\ref{fig:large-intersection}.
\end{Rule}
\begin{figure}
\centering
\includegraphics{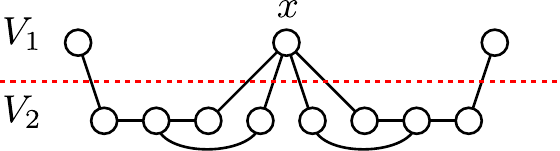}
\caption{An example for Rule~(\ref{rule:large-intersection}).}
\label{fig:large-intersection}
\end{figure}
\begin{lemma}
Rule~(\currentrule) is correct.
\end{lemma}
\begin{proof}
Since $|C \cap N(x)|\geq 2$ and $x$ splits at least one \component{2}
other than $C$ (Observation~\ref{obs:split-1}), we have that $\deg(x) \geq 3$.
By Observation~\ref{obs:P3}, $N_1(v) = \{x\}$ for every $v\in N(x)$.

Let $S$ be a \mvc of $G$.
Define $S_0 = N(x) \cup S_b$ and $S' = (S\setminus S_0) \cup \sboundary$.
The connected component of $x$ in $G-S'$ is a star.
Therefore, $S'$ is a \vc of $G$.
We will show that $|S\cap S_0| \geq |\sboundary|$
and therefore $S'$ is a \mvc of $G$.

Let $t_1$ (resp., $t_2$) be the number of \components{2} $C$
that are split by $x$ and $|C\cap N(x)| = 1$
(resp., $|C\cap N(x)| \geq 2$).
Note that $t_2 \geq 1$.
By definition, $|\sboundary| = t_1+t_2$.

Suppose first that for every \component{2} $C'$ such that $x$ splits $C'$
and $|C' \cap N(x)| = 1$ we have that $S$ contains at least one vertex from
$S_0 \cap C'$.
For every \component{2} $C'$ such that $x$ splits $C'$ and
$|C' \cap N(x)| \geq 2$, $S$ must contain at least one vertex from
$S_0 \cap C'$.
Therefore, $|S \cap S_0| \geq t_1+t_2$.

Now, suppose that there is a \component{2} $C'$ such that $x$ splits $C'$,
$|C' \cap N(x)| = 1$, and $S$ does not contain a vertex from $S_0 \cap C'$.
In this case we have $N(x)\setminus C' \subseteq S$.
Therefore, $|S\cap S_0| \geq t_1-1+2 t_2 \geq t_1+t_2$.
\end{proof}

\begin{Rule}\label{rule:split-3}
If $x$ is a connection vertex that splits at least three \components{2},
recurse on $\sboundary$ and on
$S_i = (C_i \setminus(N(x)\cup \sboundary)) \cup (N(x)\setminus C_i)$ for every
$i \leq t$,
where $C_1,\ldots,C_t$ are the \components{2} that $x$ splits.
See Figure~\ref{fig:split-3}.
\end{Rule}
\begin{figure}
\centering
\includegraphics{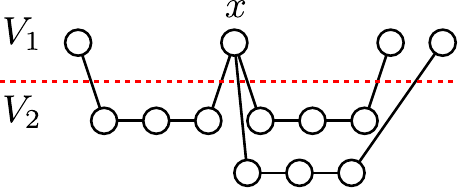}
\caption{An example for Rule~(\ref{rule:split-3}).}
\label{fig:split-3}
\end{figure}
\begin{lemma}
Rule~(\currentrule) is correct.
\end{lemma}
\begin{proof}
Let $S$ be a \mvc of $G$.
Let $S_0 = N(x) \cup \sboundary$ and $S' = (S\setminus S_0) \cup \sboundary$.
By Observation~\ref{obs:P3}, $N_1(v) = \{x\}$ for every $v\in N(x)$,
and it follows that $S'$ is a \vc of $G$.
If $S$ contains at least one vertex from $S_0 \cap C_i$ for all $i \leq t$
then $|S\cap S_0| \geq t = |\sboundary|$.
It follows that $S'$ is a \mvc of $G$.

Now suppose that there is an index $i$ such that $S$ does not contain a vertex
from $S_0 \cap C_i$. Then, $S$ must contain every vertex in
$C_i \setminus S_0$.
Additionally, $S$ must contain every vertex in $N(x) \setminus C_i$.
Therefore, $S_i \subseteq S$.
\end{proof}
The branching vector of Rule~(\currentrule) is at least $(t,t,\ldots,t)$,
where the value $t$ is repeated $t+1$ times.
The worst case is when $t = 3$, and the branching number of the vector
$(3,3,3)$ is at most 1.588.

\begin{Rule}\label{rule:large-C}
If $x$ is a connection vertex that splits the \components{2} $C$ and $C'$
such that $|C'|\geq 4$ and the vertices in $N(x)$ are not adjacent to leaves,
recurse on $S_1 = \{v,v'\}$, $S_2 = (C \setminus \{u,v\}) \cup \{u'\}$, and
$S_3 = \{u\} \cup (C' \setminus \{u',v'\})$,
where $u$ (resp., $u'$) is the single vertex in $C \cap N(x)$
(resp., $C' \cap N(x)$), and
$v$ (resp., $v'$) is the boundary vertex of $C$ (resp., $C'$) with respect to $x$.
See Figure~\ref{fig:large-C}.
\end{Rule}
\begin{figure}
\centering
\includegraphics{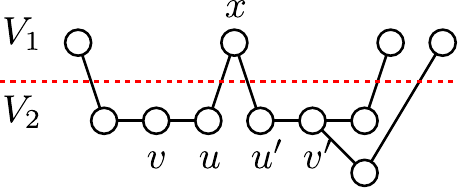}
\caption{An example for Rule~(\ref{rule:large-C}).}
\label{fig:large-C}
\end{figure}
\begin{lemma}
Rule~(\currentrule) is correct.
\end{lemma}
\begin{proof}
Let $S$ be a \mvc of $G$.
Due to the paths $x,u,v,w$ for every $w\in C \setminus \{u,v\}$,
there are three possible cases:
(1) $u \in S$,
(2) $u \notin S$ and $v\in S$
(3) $u,v \notin S$ and $C \setminus \{u,v\} \subseteq S$.

If Case~(3) occurs then $u' \in S$ (due to the path $v,u,x,u'$).
Thus, $S_2 \subseteq S$.
If Case~(2) occurs then $S$ contains either $u'$ or $v'$.
Therefore, $S' = (S\setminus \{u'\}) \cup \{v'\}$ is a \mvc of $G$
and $S_1 \subseteq S'$.

Now suppose that Case~(1) occurs.
If $u' \in S$ or $v' \in S$, the set
$S'' = (S\setminus \{u,u'\}) \cup \{v,v'\}$ is a \mvc of $G$
and $S_1 \subseteq S''$.
Otherwise ($u',v' \notin S$) we have $C' \setminus \{u',v'\} \subseteq S$.
Therefore, $S_3 \subseteq S$.
\end{proof}
The branching vector of Rule~(\currentrule) is at least $(2,2,3)$,
and the branching number is at most 1.619.

\begin{Rule}\label{rule:large-C-2}
If $C$ is a \component{2} of size at least 4,
recurse on $S_i = \{u_i\} \cup (\{v_1,\ldots,v_s\} \setminus \{v_i\})$ for
$i = 1,\ldots,s$,
where $v_1,\ldots,v_s$ are the vertices of $C$ except the center, and
$u_i$ is the unique vertex in $V_2 \setminus \{v_i\}$ that is adjacent to the
unique connection vertex that is adjacent to $v_i$.
See Figure~\ref{fig:large-C-2}.
\end{Rule}
\begin{figure}
\centering
\includegraphics{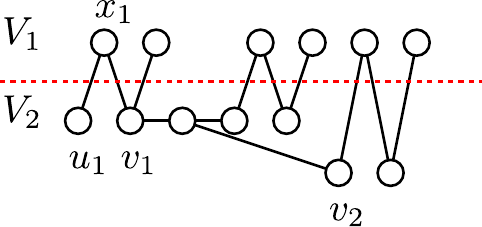}
\caption{An example for Rule~(\ref{rule:large-C-2}).}
\label{fig:large-C-2}
\end{figure}
\begin{lemma}
Rule~(\currentrule) is correct.
\end{lemma}
\begin{proof}
Due to Rule~(\ref{rule:move-to-V1}) and
Observation~\ref{obs:two-connection-vertices},
every $v_i$ is adjacent to a unique connection vertex $x_i$.
Since Rule~(\ref{rule:contains}) and Rule~(\ref{rule:split-3}) cannot be
applied, there is a unique vertex $u_i \neq v_i$ that is adjacent to $x_i$.
Therefore, the definition of Rule~(\currentrule) is valid.
Also note that the vertices $x_1,\ldots,x_s$ are distinct, due to the assumption
that Rule~(\ref{rule:large-intersection}) cannot be applied.
Since Rule~(\ref{rule:large-C}) cannot be applied, for every $i$,
at least one of the vertices $u_i$ and $v_i$ is adjacent to a leaf.

Let $S$ be a \mvc of $G$.
If there is an index $i$ such that $v_i \notin S$, then $u_i \in S$
(since at least one of the vertices $u_i$ and $v_i$ is adjacent to a leaf).
Additionally, $(\{v_1,\ldots,v_s\} \setminus \{v_i\}) \subseteq S$
(for every $j\neq i$, $S$ must contain $v_j$ due to the path $x_i,v_i,v,v_j$,
where $v$ is the center of $S$).
Therefore, $S_i \subseteq S$.
If $v_i \in S$ for all $i$, define $S' = (S\setminus \{v_1\}) \cup \{u_1\}$.
The set $S'$ is a \mvc of $G$ and $S_1 \subseteq S'$.
\end{proof}
The branching vector of Rule~(\currentrule) is $(s,s,\ldots,s)$,
where the value $s$ is repeated $s$ times.
Since $s\geq 3$, the branching number of this rule is at most 1.443.

\begin{Rule}\label{rule:last}
If none of the previous rules is applicable,
let $C$ be a connected component of $G$.
$C$ has the form
$C = C_1 \cup \cdots \cup C_s \cup \{x_1,\ldots,x_s\} \cup L$ where
(1) Each $C_i = \{u_i, v_i, w_i\}$ is a \component{2} which is a star of size~3
with a center $v_i$.
(2) Each $x_i$ is a connection vertex.
$N(x_i) = \{w_i,u_{i+1}\}$ if $i \leq s-1$ and $N(x_s) = \{w_s,u_1\}$.
(3) $L$ is a set of leaves that are adjacent to vertices in
$\{u_1,\ldots,u_s,w_1,\ldots,w_s\}$.
Recurse on $\{u_1,\ldots,u_s\}$.
See Figure~\ref{fig:last}.
\end{Rule}
\begin{figure}
\centering
\includegraphics{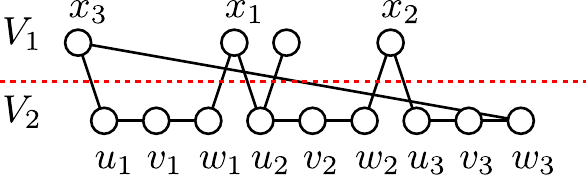}
\caption{An example for Rule~(\ref{rule:last}).}
\label{fig:last}
\end{figure}
\begin{lemma}
Rule~(\currentrule) is correct.
\end{lemma}
\begin{proof}
The set $\{u_1,\ldots,u_s\}$ is a \vc of $G[C]$.
Additionally, every \vc of $G[C]$ must contain at least one vertex from each
$C_i$.
Therefore, $\{u_1,\ldots,u_s\}$ is a \mvc of $G[C]$.
\end{proof}

Since all the branching rules have branching number at most 1.619, it follows
that the algorithm solves the disjoint 4-path vertex cover problem in
$O^*(1.619^k)$ time.
Therefore, there is an $O^*(2.619^k)$-time algorithm for 4-path vertex cover.

\bibliographystyle{plain}
\bibliography{parameterized}

\end{document}